\documentclass[11pt]{article}

\usepackage[T1]{fontenc}
\usepackage[utf8]{inputenc}
\usepackage{authblk}
\usepackage{fullpage}
\usepackage{amstext,amssymb}
\usepackage{amsmath}
\usepackage{amsthm} 
\usepackage{csquotes}
\usepackage{comment}
\usepackage{xcolor}
\usepackage{bbm}
\usepackage{mathtools}
\usepackage[margin=0.95in]{geometry}
\usepackage{ifthen}
\usepackage{algpseudocode,algorithm,algorithmicx}

\algrenewcommand\algorithmicrequire{\textbf{Precondition:}}
\algrenewcommand\algorithmicensure{\textbf{Postcondition:}}
\usepackage{graphicx,color}
\usepackage{fancybox}
\usepackage{lipsum}
\usepackage{array,float}
\usepackage{url}
\begingroup
    \makeatletter
    \@for\theoremstyle:=definition,remark,plain\do{%
        \expandafter\g@addto@macro\csname th@\theoremstyle\endcsname{%
            \addtolength\thm@preskip\parskip
            }%
        }
\endgroup
\usepackage{mathrsfs}
\usepackage{dsfont}
\definecolor{cornellred}{rgb}{0.7, 0.11, 0.11}
\definecolor{dgreen}{rgb}{0.0, 0.5, 0.0}
\definecolor{ballblue}{rgb}{0.13, 0.67, 0.8}
\definecolor{royalblue(web)}{rgb}{0.25, 0.41, 0.88}
\definecolor{bleudefrance}{rgb}{0.19, 0.55, 0.91}
\definecolor{royalazure}{rgb}{0.0, 0.22, 0.66}
\usepackage{hyperref}
\hypersetup{
    colorlinks = true,
    linkcolor=cornellred,
    citecolor=royalazure,
    linkbordercolor = {white}
}
\usepackage{cleveref}
\usepackage{enumitem}

\newtheorem{theorem}{Theorem}[section]
\newtheorem{lemma}[theorem]{Lemma}
\newtheorem*{lemma*}{Lemma}
\newtheorem{proposition}[theorem]{Proposition}

\newtheorem{observation}{Observation}
\newtheorem{claim}{Claim}
\theoremstyle{definition}

\theoremstyle{definition}

\newenvironment{remark}[1][Remark]{\begin{trivlist}
\item[\hskip \labelsep {\bfseries #1}]}{\end{trivlist}}

\usepackage{accents}
\usepackage[authoryear,round]{natbib}

\DeclareMathOperator{\OPT}{\texttt{OPT}}
\DeclareMathOperator{\MC}{\textrm{MaxCut}}
\newcommand{\SDP}{\texttt{OPT-SDP}}
\newcommand{\Ws}{\sum_{(i,j)\in E}w_{ij}}
\newcommand{\OBJ}{\texttt{OBJ}}
\newcommand{\Ycal}{\mathcal{Y}}
\newcommand{\Ecal}{\mathcal{E}}
\newcommand{\Zcal}{\mathcal{Z}}

\newcommand{\AL}{\texttt{average-linkage}}
\newcommand{\nl}{\texttt{non-leaves}}

\DeclarePairedDelimiterX{\set}[1]\{\}{#1}
\let\Pr\relax
\DeclarePairedDelimiterXPP{\Pr}[1]{\mathbb{P}}[]{}{#1}
\DeclarePairedDelimiterXPP{\Ex}[1]{\mathbb{E}}[]{}{#1}

\DeclarePairedDelimiter\floor{\lfloor}{\rfloor}

\newcommand{\tw}{\mathcal{W}}

\newcommand{\xijt}[1]{x^{(#1)}_{ij}}
\newcommand{\vit}[2]{\mathbf{v}^{(#2)}_{#1}}
\newcommand{\norms}[1]{\lVert #1 \rVert_2^2}

\DeclareFontFamily{U}{matha}{\hyphenchar\font45}
\DeclareFontShape{U}{matha}{m}{n}{
      <5> <6> <7> <8> <9> <10> gen * matha
      <10.95> matha10 <12> <14.4> <17.28> <20.74> <24.88> matha12
      }{}
\DeclareSymbolFont{matha}{U}{matha}{m}{n}
\DeclareMathSymbol{\wedge}         {2}{matha}{"5E}
\DeclareMathSymbol{\vee}           {2}{matha}{"5F}

\title{Hierarchical Clustering better than Average-Linkage}
\author[$\dagger$]{Moses Charikar}
\author[$\dagger$]{Vaggos Chatziafratis}
\author[$\dagger$]{Rad Niazadeh}
\affil[$\dagger$]{Department of Computer Science, Stanford University,}

\begin{document}
\date{}
\maketitle
\newenvironment{myfont}{\fontfamily{pag}\selectfont}{\par}

\begin{abstract}
\normalsize
Hierarchical Clustering (HC) is a widely studied problem in exploratory data analysis, usually tackled by simple agglomerative procedures like average-linkage, single-linkage or complete-linkage. In this paper we focus on two objectives, introduced recently to give insight into the performance of average-linkage clustering:
a similarity based HC objective proposed by~\cite{joshNIPS} and a dissimilarity based HC objective proposed by~\cite{vincentSODA}.
In both cases, we present tight counterexamples showing that average-linkage cannot obtain better than $\tfrac{1}{3}$ and $\tfrac{2}{3}$ approximations respectively (in the worst-case), settling an open question raised in~\cite{joshNIPS}. 
This matches the approximation ratio of a random solution,
raising a natural question: {\em can we beat average-linkage for these objectives?} We answer this in the affirmative, giving two new algorithms based on semidefinite programming with provably better guarantees. 

\end{abstract}
\section{Introduction}
Hierarchical Clustering (HC) is a popular exploratory data analysis method with a variety of applications, ranging from image and text classification to analysis of social networks and financial markets. The undisputed queen of killer applications for HC is in phylogenetics~\citep[e.g. in][]{eisen1998cluster}, where  genomic similarity (or dissimilarity) patterns are used to create taxonomies of organisms, with the goal of shedding light on the evolution of species by understanding the ancestral tree of life.

The easiest way to view HC is as a recursive partitioning of a set of datapoints into successively smaller clusters represented by a \emph{dendrogram}; a rooted tree whose leaves are in one-to-one correspondence with the datapoints (see \Cref{fig:HC}). 

HC owes its widespread success to several advantages that this tree offers compared to the more traditional ``flat'' clustering approaches like $k$-means, $k$-median or $k$-center. In fact, it does not require a fixed number $k$ of clusters and provides richer information at all levels of granularity, simultaneously displayed in an intuitive form. Importantly, there are many fast and easy to implement algorithms commonly used in practice to find the tree. Examples are simple linkage-based agglomerative procedures, with \emph{Average-Linkage} being perhaps the most popular one~\citep[e.g. see][]{friedman2001elements}.

\begin{figure}[ht]
  \centering
  \includegraphics[width=0.65\columnwidth]{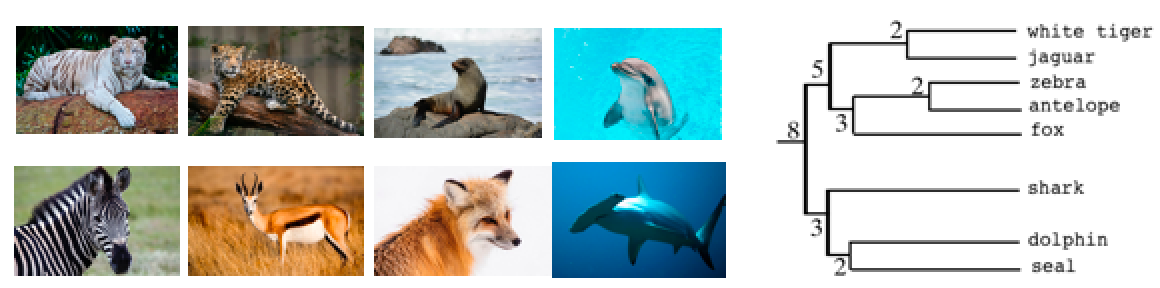}
     \caption{\footnotesize{HC dendrogram with 8 datapoints (leaves); numbers are the sizes of the clusters (tree nodes).}}
     \label{fig:HC}
\end{figure}

However, despite the immense focus on algorithms, there have been no provable guarantees accompanying their solutions, and this was partly because of the lack of objective functions to measure their qualities. To remedy this, 
\cite*{dasguptaSTOC} recently introduced and studied an interesting objective function for hierarchical clustering with a similarity measure. This objective favors cutting heavy edges deeper in the  clustering tree. More precisely, suppose $\{w_{ij}\}$ represents pairwise
similarity weights between the $n$ datapoints. \cite*{dasguptaSTOC}  proposes finding a hierarchical tree $T^*$ such that:
\begin{align}
T^*=\underset{\textrm{all trees}~T}{\mathrm{argmin}}~\sum_{(i,j)\in E}w_{ij}\cdot \lvert T_{ij}\rvert \label{obj1}\tag{\texttt{HC-OBJ-1}}
\end{align}
where $T_{ij}$ is the subtree rooted at the lowest common ancestor of $i,j$ in $T$ and $\lvert T_{ij}\rvert$ is the number of leaves it contains. He showed that solutions obtained from minimizing this objective have many desirable properties. This initiated a line of work on objective driven algorithms for HC, resulting in new algorithms as well as shedding light on the performance of classical methods. 

In particular, two recent (and independent) papers took this objective function viewpoint to understand the performance of Average-Linkage. 
In the first work, \cite*{joshNIPS} introduced a new objective that explicitly favors postponing the cutting of ``heavy'' edges to when the clusters become small, which is in some sense dual to the objective introduced by Dasgupta:
\begin{align}
T^*=\underset{\textrm{all trees}~T}{\mathrm{argmax}}~\sum_{(i,j)\in E}w_{ij}\cdot (n-\lvert T_{ij}\rvert )\label{obj2}\tag{\texttt{HC-OBJ-2}}
\end{align}
In many applications, the geometry in the data is given by dissimilarity scores instead of similarities. In the second work, 
\citealp{vincentSODA} took this view and studied a maximization version of Dasgupta's objective where pairwise weights $w_{ij}$ denote dissimilarities between the endpoints:
\begin{align}
T^*=\underset{\textrm{all trees}~T}{\mathrm{argmax}}~\sum_{(i,j)\in E}w_{ij}\cdot \lvert T_{ij}\rvert \label{obj3}\tag{\texttt{HC-OBJ-3}}
\end{align} 
For maximizing the similarity-based objective, the first work showed that Average-Linkage obtains a $\frac{1}{3}$-approximation. Interestingly, for maximizing the dissimilarity-based objective, the second work also showed that Average-Linkage gives a $\frac{2}{3}$-approximation~\citep{vincentpersonal}. Besides helping with understanding the performance of Average-Linkage, a comprehensive list of desirable properties of the aforementioned objectives can be found in~\cite{dasguptaSTOC,joshNIPS,vincentSODA}, \cite{vincentNIPS}.

\paragraph{Our Contributions.} In this paper, we shed further light on these two hierarchical clustering objectives.
Since both of the objectives were recently introduced in the context of explaining the success of Average-Linkage, and as these objectives are NP-hard to optimize~\citep{dasguptaSTOC}, it is natural to ask how well these objectives can be approximated.
Understanding the approximation factors achievable by other algorithms for these objectives is important in evaluating the explanation for the performance of Average-Linkage by these works.

It turns out that a random solution to both these optimization problems achieves an approximation ratio that matches the approximation guarantees established by the above works for Average-Linkage \citep{joshNIPS,vincentSODA}.
In particular, \citeauthor{joshNIPS} suggest that the performance of Average-Linkage may be strictly better than that of a random solution. Our first set of results is negative:

\begin{displayquote}
\emph{In the worst case, the approximation ratio achieved by Average-Linkage is no better than $\frac{1}{3}$ for the maximization objective of 
\citealp{joshNIPS} and no better than $\frac{2}{3}$ for the maximization objective of 
\citealp{vincentSODA}.}
\end{displayquote}

This raises a natural question: is it possible to achieve an approximation factor strictly better than that achieved by Average-Linkage (also by a random solution) for these two objectives? Or is it the case that these objectives are approximation resistant (i.e. beating the performance of a random solution is provably hard)? Our main result here is positive:

\begin{displayquote}
\emph{ We show simple algorithms that achieve a $\frac{1}{3}+\epsilon$ approximation for the maximization objective of \citealp{joshNIPS} and an algorithm that achieves a $\frac{2}{3}+\delta$ approximation for the maximization objective of~\citep{vincentSODA}, for constants $\epsilon,\delta>0$}
 \end
 {displayquote}
Our algorithms are conceptually very simple;
in the former case, our algorithm is guided by a semidefinite programming (SDP) solution that has a vector representation for the hierarchical clustering for every level of granularity, and uses spreading metric constraints to strengthen the solution. We use the solution at level $n/2$ to make a judicious choice of initial partition, followed by a random solution to refine each piece produced (see \cref{sec:beating-similarity} for details).
In the latter case, our algorithm follows a simple greedy peeling strategy, followed by a max-cut partition (see \cref{sec:beating-dissimilarity} for details). 

In addition to shedding light on the two objectives from the point of view of understanding their approximability, an additional lens with which to view our results is a philosophical one: our results raise the question about whether these objectives are indeed the right way to measure the performance of Average-Linkage clustering.

\paragraph{Further Related Work}
As we mentioned, there is a large body of literature on HC (we refer the reader to~\cite{berkhin2006survey} for a survey) starting with early works in phylogenetics by~\cite{sneath1962numerical,jardine1968model}.  Average-Linkage was one of several methods originating in this field that were subsequently adapted for general-purpose data analysis. Other major applications include image and text classification~\cite{steinbach2000comparison}, community detection in social networks~\cite{leskovec2014mining,mann2008use}, bioinformatics~\cite{diez2015novel}, finance~\cite{tumminello2010correlation} and more. 

Following the formulation of HC as a combinatorial optimization problem, several works explored HC from an approximation algorithms perspective. Dasgupta proposed a top-down approach based on Sparsest Cut and proved that it achieves an $O(\log^{3/2} n)$-approximation that was later improved to $O(\log n)$ by~\cite{royNIPS} and finally to $O(\sqrt{\log n})$~by \cite{vaggosSODA,vincentSODA}. Escaping from worst-case analysis,~\cite{vincentNIPS} study hierarchical extensions of the stochastic block model and show that an older spectral algorithm of~\cite{mcsherry2001spectral} augmented with linkage methods results in an $O(1)$-approximation to Dasgupta's objective.

In another line of work, hierarchical clustering in the context of ``dynamic'' or ``incremental'' clustering, using standard flat-clustering objectives like $k$-means, $k$-median or $k$-center as proxies, has been studied~(\cite{charikar2004incremental,dasgupta2002performance,plaxton2003approximation,lin2010general}).  
Furthermore, there has been recent attention on the ``semi-supervised'' or ``interactive'' versions of HC by \cite{nina1,nina2}, showing that interactive feedback in the form of cluster split/merge requests can lead to significant improvements, and by \cite{dasguptaICML,vaggosICML}, providing techniques for incorporating prior knowledge to get better hierarchical trees.

\section{Notations and Basics}

\paragraph{Notations.} We abuse notation and use $\OPT$ to refer to both the optimum solution and its value for the HC problem at hand. Similarly, $\AL$ denotes both the solution produced by the Average-Linkage algorithm and its objective value. We emphasize that when dealing with similarity weights, the $\AL$ merging criterion is to \textit{maximize} average similarity across the sub-clusters available to the algorithm for merging, whereas when dealing with dissimilarity weights, the $\AL$ merging criterion is to \textit{minimize} average dissimilarity. We term \ref{obj2} as \emph{similarity-HC} and \ref{obj3} as \emph{dissimilarity-HC} in this paper. We use $\tw$ to denote the total weight of the edges in the graph, i.e. $\tw=\sum_{e\in E}w_e$.

\section{Tight Instances for Average Linkage Analysis}
\label{sec:tight-similarity}
Here we describe the constructions of two families of examples proving that 
the known performance bounds for $\AL$ for the two objectives, i.e. \emph{similarity-HC} and \emph{dissimilarity-HC}, are tight.

\begin{observation} [\cite{dasguptaSTOC}]
\label{obs:table-vals}
If the graph is a clique with uniform weights for all of the edges, any clustering tree $T$ obtains exactly the same cost/reward. We occasionally use this fact in this section.
\end{observation}
%
%
%
%
%
%
%
%

\subsection{Average-Linkage for similarity-HC is a tight \texorpdfstring{$\frac13$-approximation}{TEXT}}

We will provide a construction where the optimum solution $\OPT$ has value $\approx n\tw$, but where $\AL$ only gets $\tfrac13 n\tw$ (ignoring lower order terms). 
%
The construction does two things:
1) Most of the graph's weight is inside subclusters containing 
$n^{2/3}$ nodes each. So there exists a solution merging almost all the weight in low levels of the hierarchical decomposition, getting 
$\approx 
n\tw$ total value. 
2) 
$\AL$ cuts most of the graph's weight in higher levels of the corresponding tree decomposition so that according to \ref{obj2}, the multiplier of the edges weights is small. 
 \begin{figure}[ht]
  \centering
  \includegraphics[width=0.65\columnwidth]{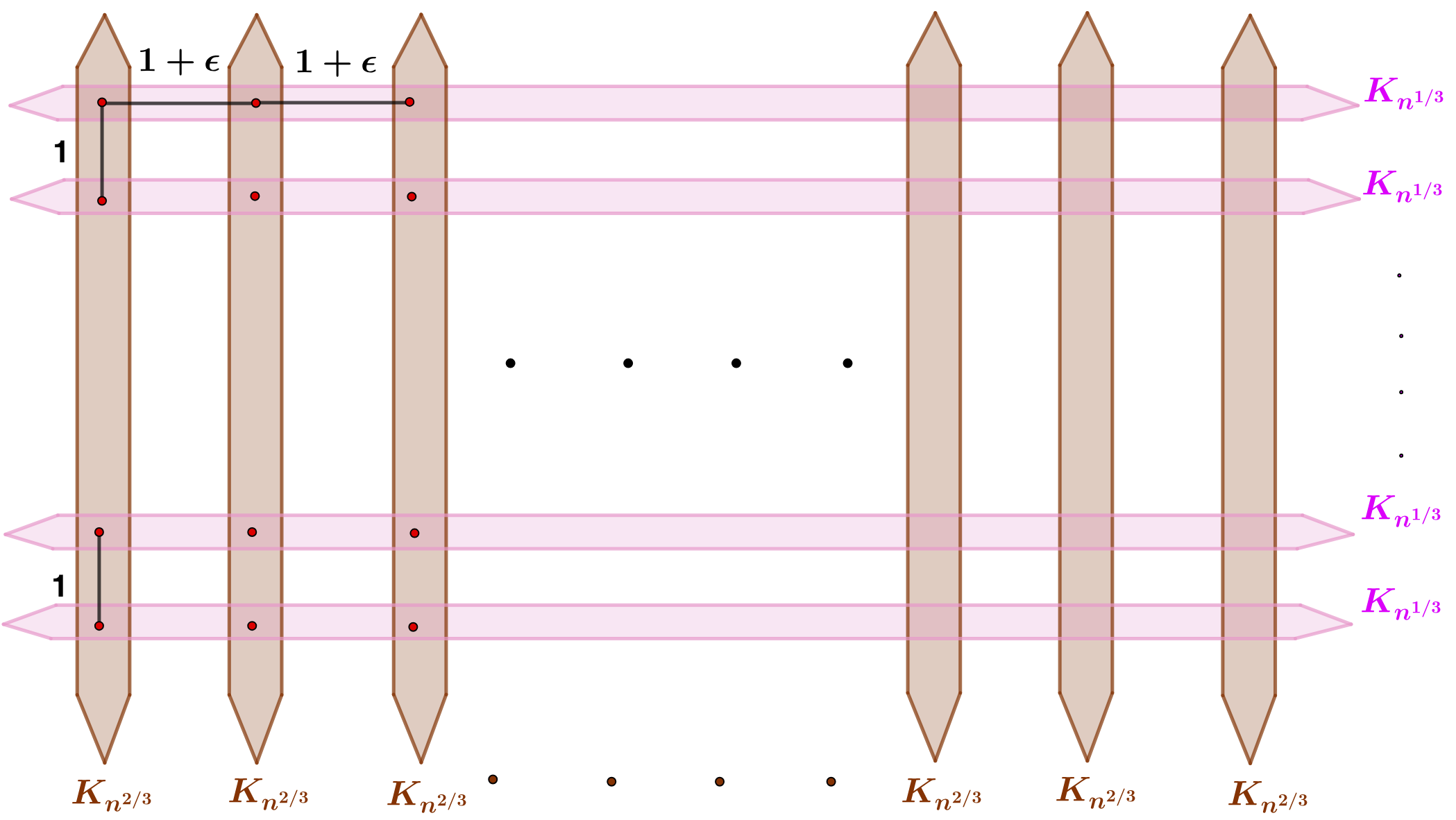}
     \caption{The tight instance where $\AL$ is a tight $\tfrac13$-approximation. The graph has $n$ nodes organized in $n^{1/3}$ vertical groups of $n^{2/3}$ vertices. Each vertical group is a clique $K_{n^{2/3}}$ on $n^{2/3}$ nodes and there are $n^{1/3}$ such groups. Each horizontal group is a clique on $n^{1/3}$ nodes. Every edge in the vertical groups has weight 1, whereas every edge in the horizontal groups has weight $1+\epsilon$.}
     \label{fig:counter1}
\end{figure}

\paragraph{Construction of the tight instance.} To achieve the above, our example (see~\hyperref[fig:counter1]{Figure~\ref{fig:counter1}}) will have $n$ nodes and will contain multiple copies of cliques each of which is either a copy of $K_{n^{1/3}}$ or $K_{n^{2/3}}$. In particular, the tight instance consists of $n^{1/3}$ copies of $K_{n^{2/3}}$ with unit weight. With a slight abuse of notation, we fix an arbitrary ordering $1,2,\dots,n^{2/3}$ for the nodes in \textit{each} $K_{n^{2/3}}$ and refer to them by their corresponding order in each clique. Now we augment this construction by adding all pairwise edges connecting nodes of the \textit{same} order across all the $K_{n^{2/3}}$ cliques. This creates $n^{2/3}$ additional $K_{n^{1/3}}$ cliques and we fix the weight of these additional edges to be $1+\epsilon$ (for any small constant $\epsilon>0$). Note that the total number of nodes is $n^{2/3}\cdot n^{1/3}=n$ and that the total weight of the graph is $\tw=\tfrac{1}{2}n^{2/3}\cdot (n^{2/3}-1)\cdot n^{1/3}\cdot 1+\tfrac{1}{2}n^{1/3}\cdot (n^{1/3}-1)\cdot n^{2/3}\cdot (1+\epsilon) =$ $\tfrac12 n^{5/3}+O(n^{4/3}).$

The following two lemmas compare $\OPT$ to $\AL$ (proofs are deferred to \Cref{appendix:sec3}).
\begin{lemma}
\label{lem:tight-sim-1}
In the above instance, the optimum obtains an objective of at least $ \tfrac12 n^{8/3} - O(n^{7/3}) \approx n\tw$.
\end{lemma}
\begin{lemma}
\label{lem:tight-sim-2}
In the above instance, Average-Linkage gets at most $\tfrac16 n^{8/3} +O(n^{7/3}) \approx \tfrac13 n\tw$.
\end{lemma}
Combining these lemmas, we settle a open question raised in~\cite{joshNIPS}:

\begin{proposition}
There exists an instance for which Average-Linkage is a $\tfrac13 +o(1)$-approximation for the similarity-HC objective~(\ref{obj2}) introduced in \cite{joshNIPS}.
\end{proposition}

\subsection{Average-Linkage for dissimilarity-HC is a tight \texorpdfstring{$\frac23$-approximation}{TEXT}}
\label{sec:dissimilarity-lowerbound}
When the pairwise scores denote dissimilarities, we focus on \ref{obj3}. \cite{vincentSODA} showed that running $\AL$ gives a $\tfrac12$-approximation, and a slight modification of their proof~\citep{vincentpersonal} gives an improved $\tfrac{2}{3}$-approximation bound. Here we show that the $\tfrac23$ ratio is actually tight.

\paragraph{Construction of the tight instance.} Let the number of nodes $n$ be even. We start with the complete bipartite graph $K_{n/2,n/2}$ with unit weights (let $L,R$ denote the two sides of the graph). We then remove any perfect matching $M$ crossing the $(L,R)$ cut (see~\hyperref[fig:counter2]{Figure~\ref{fig:counter2}}). Note that the total weight of the edges is $\tw=\tfrac n2 \cdot \tfrac n2 -\tfrac n2 \approx \tfrac14 n^2 $.
\vspace{3mm}

The following two lemmas compare $\OPT$ to $\AL$ (proofs are deferred to \Cref{appendix:sec3}).

\begin{figure}[ht]
  \centering
  \includegraphics[width=0.65\columnwidth]{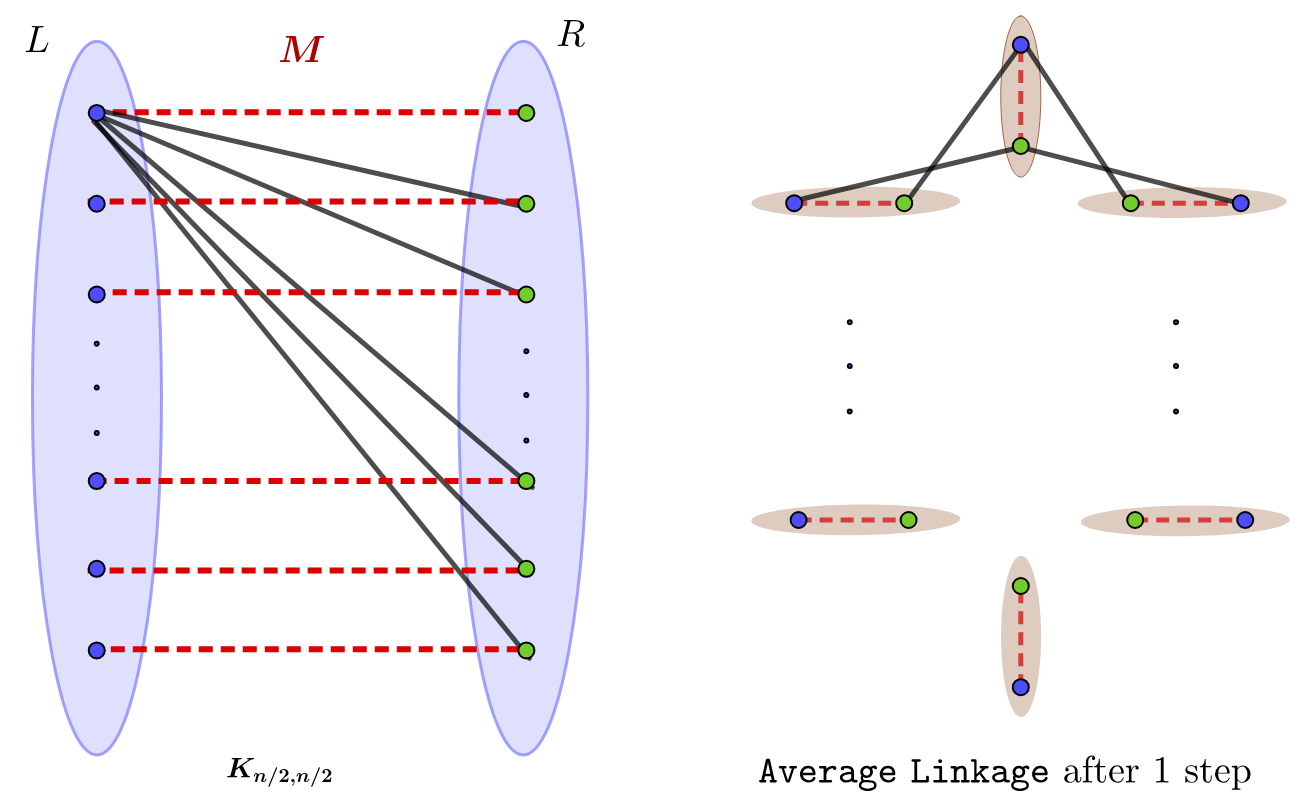}
     \caption{The tight instance where $\AL$ is a tight $\tfrac23$-approximation. The graph is a complete bipartite graph where we removed a perfect matching $M$ denoted with the red dashed edges. After one step of $\AL$, the instance is a clique on $\tfrac n2$ supernodes of size 2 with doubled edge weights.}
     \label{fig:counter2}
\end{figure}
\begin{lemma}
\label{lem:tight-disim-1}
In the above instance, the optimum HC decomposition obtains an objective value of at least $\tfrac14 n^{3} - O(n^2)\approx n\tw$.
\end{lemma}
\begin{lemma}
\label{lem:tight-disim-2}
In the above instance, Average-Linkage gets at most  $\tfrac16 n^3 \approx \tfrac23 n\tw$.
\end{lemma}
Combining these lemmas, we get the following result:
\begin{proposition}
There exists an instance for which Average-Linkage is a $\tfrac23 +o(1)$-approximation for the dissimilarity-HC objective~(\ref{obj3}), introduced by \cite{vincentSODA}.
\end{proposition}

\newcommand{\Trand}{T^{(1)}}
\newcommand{\Tsdp}{T^{(2)}}
\section{Beating Average-Linkage for Similarity-HC}
\label{sec:beating-similarity}
In this section, we aim to design approximation algorithms for the similarity-based hierarchical clustering, i.e. creating a hierarchical decomposition that approximately maximizes the objective function \ref{obj2}. To this end, we formulate a semidefinite programming (SDP) relaxation for the problem. We show that using a well chosen set of vectors returned by this SDP and a simple hyperplane rounding scheme, we can beat the average-linkage which is a tight $\tfrac13$-approximation algorithm. The main challenge in the analysis of the rounding scheme is in lower bounding the probability of certain events related to the triplets of vertices and the order in which they get separated; we do this by exploiting the specific geometry of the vectors in the SDP optimal solution, and with the help of our imposed spreading metric constraints in this SDP relaxation.
\subsection{SDP relaxation for similarity-HC}


Suppose $n$ datapoints with similarity weights $\{w_{ij}\}$ are given. An alternative view of a hierarchical clustering of these points is a collection of partitions of the points at different levels $t=n-1,\ldots,1$, where the partition at level $t$ consists of all the maximal clusters of size at most $t$. Given this view, we can rewrite the similarity-HC objective function (\ref{obj2}) as following.
\begin{equation*}
\sum_{(i,j)\in E}w_{ij}(n- \lvert T_{ij}\rvert)= \sum_{t=1}^{n-1}\sum_{(i,j)\in E}w_{ij}\cdot \mathbf{1}\{\textrm{$i$ and $j$ are not separated at level $t$'s partitioning}\}
\end{equation*}
Now, given the optimal hierarchical clustering, consider a vector assignment where at every level $t=1,..,n-1$ the same unit vectors are assigned to all the nodes in the same maximal cluster, while the assigned vectors to different clusters are chosen to be orthogonal. Let $\{\vit{i}{t}\}$ be the set of assigned vectors. Clearly, the contribution of an edge $w_{ij}$ at level $t$ can be alternatively written as $w_{ij}(\vit{i}{t}\cdot\vit{j}{t})$. This observation suggests a relaxation through semidefinite programming/vector programming.
\begin{proposition}
The following SDP is a relaxation for the similarity-HC problem~(\ref{obj2}).
\begin{equation}
\tag{HC-SDP}
\label{eq:SDP-HC}
\begin{aligned}
& \underset{}{\text{maximize}}
& & \sum_{t=1}^{n-1}\sum_{(i,j)\in E}w_{ij}(1-\xijt{t}) \\
& \text{subject to}
& & \xijt{t}=\frac{1}{2}\norms{\vit{i}{t}-\vit{j}{t}}, & \forall(i,j) \in E, t\in[1:n-1]\\
& 
& & \sum_{j\in V: j\neq i}\xijt{t}\geq n-t, &\forall i\in V ,t\in[1:n-1]&~~~~~~(\textit{spreading constraints})\\
&
& & \xijt{t+1}\leq \xijt{t},~~\xijt{1}=1& \forall(i,j)\in E,t\in[1:n-1]&~~~~~~(\textit{monotonicity constraints})\\
&
& & \vit{i}{t}\in\mathbb{R}^n,~~\norms{\vit{i}{t}}=1, & \forall i\in V,t\in[1:n-1]
\end{aligned}
\end{equation}
\end{proposition}

In an integral HC decomposition, each node in a maximal cluster of size at most $t$ has been separated from at least $n-t$ vertices at level $t$ (i.e. all of the nodes outside of this cluster). We therefore add the constraints $\sum_{j}\xijt{t}\geq n-t$, termed as the \textit{spreading constraints}. Intuitively, they force the SDP to choose vectors that are somewhat separated, thus preventing it from cheating by assigning identical vectors. Finally, \emph{monotonocity constraints} ensure monotonicty of the separation from top to bottom.

\subsection{Combining SDP rounding and random to beat average-linkage}
Suppose $\SDP$ and $\OPT$ denote the optimum solution of the SDP relaxation (\ref{eq:SDP-HC}) and the optimum integral solution of the similarity-HC objective (\ref{obj2}) respectively. Our goal is to beat the $\tfrac13$ approximation ratio attained by $\AL$. We will consider two simple algorithms for the HC problem. The first algorithm, ``\emph{random always}'', cuts each cluster recursively and uniformly at random until it reaches to singletons. The second algorithm, ``\emph{SDP first, random next}'', uses the semidefinite program \ref{eq:SDP-HC} and hyperplane rounding for determining the first cut, and then it picks a random cut for each of the later clusters until it reaches to singletons. 

As it is also known~(\cite{joshNIPS}), recursively performing random cuts will yield an HC solution with expected value exactly equal to $\tfrac13(n-2)\tw$, where $\tw\triangleq \sum_{(i,j)\in E}w_{ij}$. An initial idea is that when there is a gap between $\OPT$ and the quantity $(n-2)\tw$, i.e. when $\OPT < (1-\epsilon_1)(n-2)\tw$ for some small constant $\epsilon_1>0$, then ``\emph{random always}'' already attains an approximation guarantee of $\tfrac{1}{3(1-\epsilon_1)}>\tfrac13$. The name of the game is then to come up with a good approximation algorithm in the case where $\OPT$ is actually pretty large, close to $(n-2)\tw$. Interestingly, a suitable initial cut can be found by exploiting the SDP relaxation in this case, which can be then used to guide the ``\emph{random always}'' algorithm.

\begin{algorithm}[h]
\caption{Random Always}
\label{alg:Random}
\begin{algorithmic}[1]
\State \textbf{input}: $G=(V,E)$.
\If {$\lvert V \rvert=1$} 
\State Return the singleton vertex as the only cluster.
\Else
\State Randomly partition the set of vertices into $S$ and $\bar{S}$. 
\State Recursively run ``\emph{random always}'' on $G_S$ and $G_{\bar{S}}$ to get clusters $\mathcal{C}_S$ and $\mathcal{C}_{\bar{S}}$.
\State Return the clusters $S$,$\bar{S}$, $\mathcal{C}_S$ and $\mathcal{C}_{\bar{S}}$. 
\EndIf
\end{algorithmic}
\end{algorithm}
\begin{algorithm}[h]
\caption{SDP First, Random Next}
\label{alg:SDP-Random}
\begin{algorithmic}[1]
\State \textbf{input:} $G=(V,E)$ and (similarity) weights $\{w_{ij}\}_{(i,j)\in E}$.
\State Solve the SDP relaxation \ref{eq:SDP-HC} to get an optimum assignment $\{x^t_{ij}\}_{(i,j)\in E,~t=1,\ldots,n-1}$. 
\State Let $x^*_{ij}=\xijt{\floor{n/2}-1}$ and $\mathbf{v}^*_i=\vit{i}{\floor{n/2}-1}$ be the optimal solution restricted to level $t=\floor{n/2}-1$.
\State Draw $\mathbf{v}_0$ uniformly at random from unit sphere, and let $S=\{i\in V: \mathbf{v}^*_i\cdot \mathbf{v}_0\geq 0\}$.
\State Partition the vertices into $S$ and $\bar{S}=V\setminus S$. 
\State Run ``\emph{random always}'' (\Cref{alg:Random}) on $S$ and $\bar{S}$ to get clusters $\mathcal{C}_S$ and $\mathcal{C}_{\bar{S}}$.
\State Return the clusters $S$,$\bar{S}$, $\mathcal{C}_S$ and $\mathcal{C}_{\bar{S}}$. 
\end{algorithmic}
\end{algorithm}

\begin{theorem}
\label{thm:similarity}
The best of the ``\emph{SDP first, random next}'' (\Cref{alg:SDP-Random}) and ``\emph{random always}'' (\Cref{alg:Random}) is a randomized $\alpha$-approximation algorithm for maximizing the similarity-HC objective for hierarchical clustering, where $\alpha=0.336379>\tfrac{1}{3}.$
\end{theorem}
\subsection{Analysis (Proof of \texorpdfstring{\Cref{thm:similarity}}{TEXT})}

We start by decomposing the similarity-HC objective as a summation over contributions of different triplets of vertices $i,j$ and $k$, where $(i,j)\in E$ and $k\neq i,j$. Accordingly, \ref{obj2} can be rewritten as:
\begin{equation}
\sum_{(i,j)\in E}w_{ij}\left(n-\lvert T_{ij}\rvert\right)=\sum_{(i,j)\in E}w_{ij}\lvert\nl(T_{ij})\rvert=\sum_{(i,j)\in E}\sum_{k\neq i,j}w_{ij}\mathbf{1}\{\textrm{$k$ is not a leaf of $T_{ij}$}\}
\end{equation}
The vertex $k$ does \emph{not} belong to the leaves of $T_{ij}$ if and only if at some point during the execution of the algorithm, $k$ gets separated from $i$ and $j$, while $i$ and $j$ still remain in the same cluster. Suppose $\Trand$ and $\Tsdp$ denote the HC tree returned by \Cref{alg:Random} and \Cref{alg:SDP-Random}. Moreover, let the random variables $\Zcal_{i,j,k}$ and $\Ycal_{i,j,k}$ denote the contributions of the edge $(i,j)$ and vertex $k\neq i,j$ to the objective value of \Cref{alg:Random} and \Cref{alg:SDP-Random} respectively, i.e.,
$$\Zcal_{i,j,k}\triangleq w_{ij}\mathbf{1}\{\textrm{$k$ is a non-leaf of $\Trand_{ij}$}\}~~~\textrm{and}~~~ \Ycal_{i,j,k}\triangleq w_{ij}\mathbf{1}\{\textrm{$k$ is a non-leaf of $\Tsdp_{ij}$}\}$$
Moreover, let $\Ycal_{i,j}=\sum_{k\neq i,j}\Ycal_{i,j,k}$ and $\Zcal_{i,j}=\sum_{k\neq i,j}\Zcal_{i,j,k}$. Let $\OPT$ be the optimal value of the HC objective. Fix $\epsilon_1>0$ and consider two cases. 

\paragraph{\textbf{Case 1:}}  $\OPT < (1-\epsilon_1)(n-2)\Ws$. By a simple argument, we claim that $\Ex{\Zcal_{i,j,k}}=\frac{1}{3}w_{ij}$. Given this claim, the expected objective value of  \Cref{alg:Random} is at least $\frac{n-2}{3}\sum_{(i,j)\in E}w_{ij}$. Moreover, $\OPT\leq (n-2)\Ws$, and hence \Cref{alg:Random} obtains $\frac{1}{3(1-\epsilon_1)}$ fraction of $\OPT$ in this case. To see why the claim holds, think of each random cut as flipping an independent unbiased coin for each vertex, and then placing the vertex on either sides of the cut based on the outcome of its coin. Now, look at the sequence of the coin flips of $i$, $j$ and $k$ during the execution of \Cref{alg:Random}. We want to find the probability of the event that for the first time the three sequences are not matched, but $i$'s sequence and $j$'s sequence are still matched. Fixing $i$'s sequence, the probability that all three are always matched is $\sum_{i=1}^{\infty}(1/4)^i=1/3$. Due to the symmetry, the rest of the probability will be divided equally between our target event and the event that for the first time these three sequences are not matched, but still $i$'s sequence and $k$'s sequence are matched. So, $\Pr{\textrm{$k$ is not a leaf of $\Trand_{ij}$}}=1/3$, which proves the claim.

\paragraph{\textbf{Case 2:}}  $\OPT \geq (n-2)(1-\epsilon_1)\Ws$. In this case, we want to find a lower bound on the objective value of Algorithm~\ref{alg:SDP-Random}. To this end, we show how to bound $\Ex{\Ycal_{i,j}}$ from below for a large enough fraction of edge weights. Consider the following events:
\begin{align*}
\Ecal_{i,j}&\triangleq \{\textrm{$i$ and $j$ remain together after the first cut}\},\\
\Ecal_{i,j,k}&\triangleq \{\textrm{$i$, $j$ and $k$  remain together after the first cut}\},\\
{\Ecal}_{i,j|k}&\triangleq \{\textrm{$i$, $j$ remain together and $k$  gets separated after the first cut}\}
\end{align*}
 We can rewrite $\Ex{\Ycal_{i,j,k}}$ as follows.
\begin{align*}
\Ex{\Ycal_{i,j,k}}&=\Ex{\Ycal_{i,j}\mathbf{1}\{\Ecal_{i,j}\}}=\Ex{\Ycal_{i,j,k}\mathbf{1}\{\Ecal_{i,j,k}\}}+\Ex{\Ycal_{i,j,k}\mathbf{1}\{{\Ecal}_{i,j|k}\}}\\
&=\Ex{\Ycal_{i,j,k}|\Ecal_{i,j,k}}\Pr{\Ecal_{i,j,k}}+\Ex{\Ycal_{i,j,k}|\Ecal_{i,j|k}}\Pr{\Ecal_{i,j|k}}
\end{align*}
Now, note that $\Ex{\Ycal_{i,j,k}|\Ecal_{i,j|k}}=w_{ij}$, as $k$ has been separated from $i$ and $j$ after the first cut.  Moreover, $\Pr{\textrm{$k$ is a non-leaf of $\Tsdp_{ij}$}|\Ecal_{i,j,k}}=1/3$, as Algorithm~\ref{alg:SDP-Random} performs random cuts after the first cut and the previous argument for analyzing $\Zcal_{i,j,k}$ will be applied. So, $\Ex{\Ycal_{i,j,k}|\Ecal_{i,j,k}}=w_{ij}/3$. Therefore, we have:
\begin{align*}
\Ex{\Ycal_{i,j,k}}&=\frac{w_{ij}}{3}\Pr{\Ecal_{i,j,k}}+w_{ij}\Pr{\Ecal_{i,j|k}}=\frac{w_{ij}}{3}\left(\Pr{\Ecal_{i,j,k}}+\Pr{\Ecal_{i,j | k}}+2\Pr{\Ecal_{i,j|k}}\right)=\frac{w_{ij}}{3}\left(\Pr{\Ecal_{i,j}}+2\Pr{\Ecal_{i,j|k}}\right),
\end{align*}
and hence we have:
\begin{equation}
\label{eq:yijk}
\Ex{\Ycal_{i,j}}=\sum_{k\neq i,j}\Ex{\Ycal_{i,j,k}}=\frac{w_{ij}}{3}\left((n-2)\Pr{\Ecal_{i,j}}+2\sum_{k\neq i,j}\Pr{{\Ecal}_{i,j|k}}\right)
\end{equation}
Fix $\epsilon_2>0$. Consider all edges for which $x^*_{ij}=\xijt{\floor{n/2}-1}\leq \epsilon_2$ (denoted by $\mathcal{H}\subseteq E$).  For each $(i,j)\in\mathcal{H}$, by applying the basics of hyperplane rounding, e.g. in \cite{goemans1995improved}, we have:
\begin{equation*}
\Pr{\Ecal_{i,j}}=\Pr{(\mathbf{v}^*_i\cdot \mathbf{v}_0)(\mathbf{v}^*_j\cdot \mathbf{v}_0)\geq 0}=1-\theta_{ij}/\pi,
\end{equation*}
where $\theta_{ij}$ is defined to be the angle between the vectors $\mathbf{v}^*_i$ and $\mathbf{v}^*_j$, i.e. $\theta_{i,j}=\cos^{-1}(\mathbf{v}^*_i\cdot\mathbf{v}^*_j)=\cos^{-1}(1-x^*_{ij})$. Now, it is clear that $\theta_{ij}\leq \bar{\theta}$ for edges in $\mathcal{H}$, where $1-\cos(\bar{\theta})=\epsilon_2$. Therefore,
\begin{equation}
\label{eq:h-lower}
\Pr{\Ecal_{i,j}}\geq 1-\bar{\theta}/\pi,~~~\forall (i,j)\in \mathcal{H}
\end{equation}

To bound $\sum_{k\neq i,j}\Pr{{\Ecal}_{i,j|k}}$ for every edge $(i,j)\in\mathcal{H}$, we first find an explicit closed-form for each probability term. Interestingly, despite the complicated nature of this calculation, our method is simple and is not relaying on any three-dimensional geometry. Hence, it might be of independent interest.

\begin{remark} To calculate an explicit closed-form for the probability of the event $\Ecal_{i,j|k}$, three involved correlated random variables $\mathbf{v}^*_i\cdot\mathbf{v}_0$, $\mathbf{v}^*_j\cdot\mathbf{v}_0$ and $\mathbf{v}^*_k\cdot\mathbf{v}_0$ need to be considered. In the direct approach, e.g. à la \cite{goemans1995improved}, we need to look at the unit projections of these three vectors and $\mathbf{v}_0$ onto the span of $\mathbf{v}^*_i$, $\mathbf{v}^*_j$, and $\mathbf{v}^*_k$ (hence a three-dimensional representation for each). Suppose $\mathbf{\tilde{v}}_0$ be the projection of $\mathbf{v}_0$ onto the mentioned three-dimensional space. As the entries of $\mathbf{v}_0$ are jointly Gaussian, $\mathbf{\tilde{v}}_0$ is indeed a uniformly random point from the three-dimensional sphere. Now, finding the probability of the event that $i$ and $j$ are on one side and $k$ is on the other side of the hyperplane with normal vector $\mathbf{\tilde{v}}_0$ involves a complicated calculation in this three dimensional geometry. 
\end{remark}
\begin{lemma} 
\label{lem:closed-form-prob}
For every triplet of vertices $i$,$j$ and $k$, $\Pr{{\Ecal}_{i,j|k}}=\frac{\theta_{ik}+\theta_{jk}-\theta_{ij}}{2\pi}$
\end{lemma}
\begin{proof}
We start by the following key observation, which relates the quantities $\Pr{{\Ecal}_{i,j|k}}$, $\Pr{{\Ecal}_{i,k|j}}$ and $\Pr{{\Ecal}_{k,j|i}}$ to the original separation probabilities of a hyperplane rounding scheme:
\begin{enumerate}
\item $\Pr{{\Ecal}_{i,k|j}}+\Pr{{\Ecal}_{j,k|i}}=1-\Pr{{\Ecal}_{i,j}}=\tfrac{\theta_{ij}}{\pi}$
\item $\Pr{{\Ecal}_{i,j|k}}+\Pr{{\Ecal}_{i,k|j}}=1-\Pr{{\Ecal}_{j,k}}=\tfrac{\theta_{jk}}{\pi}$
\item $\Pr{{\Ecal}_{j,k|i}}+\Pr{{\Ecal}_{i,j|k}}=1-\Pr{{\Ecal}_{i,k}}=\tfrac{\theta_{ik}}{\pi}$
\end{enumerate}
Solving the above $3\times3$ system, we can obtain the desired closed-form expressions for $\Pr{{\Ecal}_{i,j|k}}$, $\Pr{{\Ecal}_{i,k|j}}$ and $\Pr{{\Ecal}_{k,j|i}}$ in terms of the angles between the vectors:
\begin{equation*}
\Pr{{\Ecal}_{i,j|k}}= \tfrac{\theta_{ik}+\theta_{jk} - \theta_{ij}}{2\pi}~~~,~~~\Pr{{\Ecal}_{i,k|j}}= \tfrac{\theta_{ij}+\theta_{jk} - \theta_{ik}}{2\pi}~~~,~~~\Pr{{\Ecal}_{k,j|i}}= \tfrac{\theta_{ik}+\theta_{ij} - \theta_{jk}}{2\pi}\qedhere
\end{equation*}

\end{proof}

%
%
%
%
%
%

In the next step of the analysis, we find a worst-case lower-bound for $\sum_{k\neq i,j}\Pr{{\Ecal}_{i,j|k}}$ through a \emph{factor revealing program}. More accurately, we set up a minimization problem where the objective function is equal to $\sum_{k\neq i,j}\Pr{{\Ecal}_{i,j|k}}$. For the constraints, note that due to the spreading constraints of the similarity-HC SDP relaxation (\ref{eq:SDP-HC}) for vertices $i$ and $j$, we have 
$$\sum_{k\neq i}{x^*_{ik}}\geq n-\floor{n/2}+1~~~~,~~~~\sum_{k\neq j}{x^*_{jk}}\geq n-\floor{n/2}+1$$
and therefore $\sum_{k\neq i}\cos(\theta_{ik})\leq n/2-1$ and $\sum_{k\neq j}\cos(\theta_{jk})\leq n/2-1$. Now, by applying \Cref{lem:closed-form-prob}, we can lower bound $\sum_{k\neq i,j}\Pr{{\Ecal}_{i,j|k}}$ by the optimal solution of the following optimization problem:

\begin{equation}
\label{eq:factor-reveal}
\tag{$\mathcal{P}_{\textrm{lower-bound}}$}
\begin{aligned}
& \underset{}{\textit{minimize}}
& & \frac{1}{2\pi}\sum_{k\neq i,j}  \left(\theta_{ik}+\theta_{jk} - \bar{\theta}\right) \\
& \textit{subject to}
& &\sum_{k\neq i}\cos(\theta_{ik})\leq n/2-1, & \\
&
& &\sum_{k\neq j}\cos(\theta_{jk})\leq n/2-1, & \\
& 
& &0\leq \theta_{ik}\leq\frac{\pi}{2}~~,~~0\leq \theta_{jk}\leq\frac{\pi}{2} &~~~\forall k\\
\end{aligned}
\end{equation}
Note that we restrict our attention to $0\leq\theta_{ik},\theta_{jk}\leq \pi/2$, simply because in \ref{eq:SDP-HC} we force $\xijt{t}\leq 1$, and hence $\mathbf{v}^*_i\cdot\mathbf{v}^*_j\geq 0$ for all $i,j\in V$. We now have this lemma, whose proof is deferred to \Cref{appendix:sec4}.
\begin{lemma} 
\label{lem:factor-revealing}
The optimal solution of the factor revealing program \ref{eq:factor-reveal} is lower-bounded by 
$$
(n-2)\left(\tfrac{1}{4}-\tfrac{\bar{\theta}}{2\pi}\right)
$$
\end{lemma}


By combining \cref{eq:yijk} and \cref{eq:h-lower} with \Cref{lem:factor-revealing}, we have:
\begin{align}
\Ex{\OBJ_{\textrm{ALG}_2}}&\geq \sum_{(i,j)\in\mathcal{H}}\Ex{\Ycal_{i,j}}\geq (n-2)\left(\sum_{(i,j)\in\mathcal{H}}w_{ij}\right) \left(\frac{1}{3}\cdot \left(1-\frac{\bar{\theta}}{\pi}\right)+ \frac{2}{3}\cdot \left(\frac{1}{4}-\frac{\bar{\theta}}{2\pi}\right)\right)\nonumber \\
\label{eq:bound-sdp}
&\geq 
(n-2)\left(\sum_{(i,j)\in\mathcal{H}}w_{ij}\right)\left(\frac{1}{2}-\frac{2\bar{\theta}}{3\pi} \right)
\end{align}
We finally bound the total weight of edges in $\mathcal{H}$. Note that for an edge $(i,j)\notin\mathcal{H}$, $x^*_{ij}=\xijt{\floor{n/2}-1}>\epsilon_2$. Therefore,  due to the monotonicity constraint in \ref{eq:SDP-HC}, $\xijt{t}>\epsilon_2,\forall 1\leq t\leq \floor{n/2}-1$. Now we have:
\begin{align*}
\textrm{\SDP}&=\sum_{t=1}^{n-1}\sum_{(i,j)\in E}w_{ij}(1-\xijt{t})\leq(n-2)\sum_{(i,j)\in E}w_{ij}-\sum_{t=1}^{\floor{n/2}-1}\sum_{(i,j)\in E}w_{ij}\xijt{t}\\
& \leq (n-2)\sum_{(i,j)\in E}w_{ij}-\frac{\epsilon_2(n-2)}{2}\cdot\sum_{(i,j)\in E\setminus \mathcal{H}}w_{ij}
\end{align*}
On the other hand, based on the assumption of Case 2, we know 
\begin{equation}
\SDP\geq\OPT\geq (n-2)(1-\epsilon_1)\sum_{(i,j)\in E}w_{ij}
\end{equation}
By rearranging the terms we have $\sum_{(i,j)\in \mathcal{H}}w_{ij}\geq (1-\frac{2\epsilon_1}{\epsilon_2})\Ws$. Now, combined with \Cref{eq:bound-sdp}, we can show \Cref{alg:SDP-Random} obtains $(1-\frac{2\epsilon_1}{\epsilon_2})(\frac{1}{2}-\frac{2\cos^{-1}(1-\epsilon_2)}{3\pi})$ fraction of $\OPT$. 

By balancing out the two cases, finding the optimal $\epsilon_1$ as a function of $\epsilon_2$, and finally by setting $\epsilon_2\approx 0.139$ we get the desired approximation factor of $\approx0.336379$. For more details, refer to \Cref{appendix:sec4}.

\section{Beating Average-Linkage for Dissimilarity-HC}
\label{sec:beating-dissimilarity}
\newcommand{\algpeel}{\texttt{ALG}_{\textrm{peel}}}
\newcommand{\algtotal}{\texttt{ALG}}
\newcommand{\algnpeel}{\texttt{ALG}_{\textrm{cut}}}
\newcommand{\optpeel}{\texttt{OPT}_{\textrm{red}}}
\newcommand{\opttotal}{\texttt{OPT}}
\newcommand{\optnpeel}{\texttt{OPT}_{\textrm{blue}}}
\newcommand{\optnpeelchain}{\texttt{OPT}_{\textrm{blue-chain}}}
\newcommand{\optnpeelcut}{\texttt{OPT}_{\textrm{blue-cut}}}
In this section we focus on the dissimilarity-HC objective~\eqref{obj3}. As demonstrated in \Cref{sec:dissimilarity-lowerbound}, $\AL$ fails to obtain better than $\frac{2}{3}$ fraction of the optimum in the worst-case. Similarly, ``\emph{random always}''  fails to beat this approximation ratio, simply because its objective value on any instance is exactly equal to $\frac{2}{3}n\tw$, while in a bipartite graph the optimum dissimilarity-HC objective is equal to $n\tw$. Therefore, one natural question to ask is if there exists a polynomial time algorithm that can beat the $\frac{2}{3}$ approximation factor. We answer this question in the affirmative by providing a simple algorithm. 
\subsection{The ``Peel-off First, Max-cut Next'' Algorithm}
By  looking at the structure of the dissimilarity-HC objective function in \cref{obj3}, it is clear that the top-level cuts of the tree, i.e. those corresponding to clusters of larger sizes, have considerable contributions to the objective function. For example, consider a simple algorithm that starts with a random cut and then forms the rest of the tree arbitrarily. This algorithm can still obtain an objective value of $\frac{1}{2}n\tw$. Inspired by this observation, a tempting idea to beat the approximation factor of $\frac{2}{3}$ is to start with an approximation algorithm for the max-cut, e.g. \cite{goemans1995improved}, and then construct the rest of the hierarchical tree (probably by random cutting or by continuing with the same max-cut algorithm). 

However, this naive approach fails because of the following instance. Suppose we have a graph with an embedded clique of size $\epsilon n$ (for an arbitrarily small $\epsilon>0$) and the rest of the weights are zero. The optimum dissimilarity-HC solution clearly peels off vertices of the clique one by one, and obtains an objective value of at least $n(1-\epsilon)\tw$. However, the ``recursive max-cut'' or the ``max-cut first, random next'' both cut the clique into two (almost) symmetric halves at each iteration, and obtain an objective value of at most 
$$\textrm{objective-value}\leq \left(n\frac{\tw}{2}+\epsilon n\frac{\tw}{4}\right)+\left(n\frac{\tw}{8}+\epsilon n\frac{\tw}{16}\right)+\left(n\frac{\tw}{32}+\epsilon n \frac{\tw}{64}\right)+\ldots\leq \frac{2+\epsilon}{3}n\tw$$  

The above example suggests a natural modification to our idea, i.e. to first peel off \emph{high weighted degree} vertices, and then use a max-cut algorithm. Intuitively, if in such a pre-processed instance the optimum objective value of the dissimilarity-HC is close to $n\tw$, then there should exist a considerably large cut. This large cut can be detected by an approximate max-cut algorithm, and will provide a large enough objective value for the dissimilarity-HC if used as a top-level cut in the final hierarchical clustering tree. If there is a constant gap between the optimum and $n\tw$, one can run ``\emph{random always}''  and already get an approximation factor strictly better than $\frac{2}{3}$. Formally, we propose \emph{``peel-off first, max-cut next''} (\Cref{alg:peel-maxcut}) and show how the better of this algorithm and the ``\emph{random always}''  (\Cref{alg:Random}) beats the $\frac{2}{3}$-approximation factor by a small constant. 

\begin{algorithm}[h]
\caption{Peel-off First, Max-cut Next}
\label{alg:peel-maxcut}
\begin{algorithmic}[1]
\State \textbf{input:} $G=(V,E)$, (dissimilarity) weights $\{w_{ij}\}_{(i,j)\in E}$, and parameter $\gamma>0$. 
\State Initialize hierarchical clustering tree $T\leftarrow \emptyset$.
\State Set the \emph{peeling-off threshold} to be $\tau=\frac{2\tw}{n}\cdot \gamma$. 
\State $\tilde{V}\leftarrow V$ and $\tilde{E}\leftarrow E$. 
\While{$\exists$ a vertex $v\in\tilde{V}$ such that $\displaystyle\sum_{u\in V: (v,u)\in E}w_{vu}>\tau$}\Comment \texttt{peeling off phase.}
\State Update the HC binary tree $T$ by adding the cut $(\{v\},\tilde{V}\setminus\{v\})$ to the tree.
\State $\tilde{V}\leftarrow\tilde{V}\setminus \{v\}$ and $\tilde{E}\leftarrow \tilde{E}\setminus \{e\in E: e~\textrm{incident to}~v \}$.\Comment {\texttt{induced subgraph on $\tilde{V}\setminus \{v\}$.}}
\EndWhile
\State Run \citet{goemans1995improved} for max-cut on $\tilde{G}=(\tilde{V},\tilde{E})$. \Comment{\texttt{max-cutting phase.}}
\State Let the resulting cut be $(S,\tilde{V}\setminus S)$, and update the HC tree $T$ by adding this cut to the tree.
\State Run ``\emph{random always}'' (\Cref{alg:Random}) on $S$ and $\tilde{V}\setminus S$. Add the resulting binary trees to $T$. 
\State Return the tree $T$. 
\end{algorithmic}
\end{algorithm}

\begin{theorem}
\label{thm:vincent-upperbound}
There exists a choice of $\gamma>0$ so that the best of \emph{``peel-off first, max-cut next''} with parameter $\gamma$ (\Cref{alg:peel-maxcut})  and ``\emph{random always}''  (Algorithm~\ref{alg:Random}) is an $\alpha$-approximation algorithm for maximizing the dissimilarity-HC objective, where $\alpha=0.667078 > \tfrac23$.
\end{theorem}

\subsection{Analysis (Proof of \texorpdfstring{\Cref{thm:vincent-upperbound}}{TEXT})}

Fix a parameter $\epsilon$. Let $\OPT$ be the optimal objective value of the dissimilarity-HC. Similar to the proof of \Cref{thm:similarity}, consider two cases:
\paragraph{\textbf{Case 1:}}  $\OPT < (1-\epsilon)n\tw$. A simple argument  (refer to the proof of \Cref{thm:similarity}) shows that the expected objective value of \Cref{alg:Random} is exactly equal to $\frac{2}{3}{n\tw}$ in this case, and therefore it obtains $\frac{2}{3(1-\epsilon)}$ fraction of $\OPT$ (for an exposition of this proof, we refer the reader to \citealp{chatziafratis2018hierarchical})

\paragraph{\textbf{Case 2:}}  $\OPT \geq (1-\epsilon)n\tw$.
In this case, let the optimum (binary) HC tree be $T^*$. Fix another parameter $\delta <\tfrac12$. The collection of all maximal clusters of size at most $n(1-\delta)$ forms a partition of the vertices. Here is a recursive way of looking at this partition: Imagine we start from the root of $T^*$. Each time the optimum tree performs a binary cut, we consider the two produced clusters (see \Cref{fig:opt-structure}). If any of these clusters has size at most $n(1-\delta)$, then it is a maximal cluster of size at most $n(1-\delta)$ and, by definition, it will be added to the partition. As $\delta<\frac12$, either both of these clusters must be of size at most $n(1-\delta)$, or exactly one of them is smaller than $n(1-\delta)$ while the other is strictly larger than $n(1-\delta)$. If the latter is true, we recursively follow the tree along the bigger cluster, i.e. we make the bigger cluster the new root and we iterate. Otherwise, if the former is true, we stop following the tree as we would have already produced two smaller than $n(1-\delta)$ pieces.  We denote the produced sequence of clusters by $(L_1,R_1), \ldots,(L_k,R_k)$, where $(L_i,R_i)$ are the two clusters produced by $T^*$ at the $i^{\textrm{th}}$ split. Without loss of generality we set:

$$\lvert L_i \rvert\ge n(1-\delta)> \lvert R_i \rvert,~i=1,\ldots,k-1~~~~~~~~\textrm{and}~~~~~~~~\max\left(\lvert L_k\rvert, \lvert R_k\rvert\right)< n(1-\delta)$$
 Based on the above construction, the resulting partition consists of the sets $R_1, R_2,\ldots, R_k$ and $L_k$. Note that $\lvert L_k\rvert+\lvert R_k\rvert=\lvert L_{k-1}\rvert\geq n(1-\delta)$, and therefore the rest of the graph contains $\rvert V\setminus(L_k\cup R_k)\rvert = n - \lvert L_{k-1}\rvert \leq \delta n$ vertices. 
 
 \begin{figure}[ht]
  \centering
  \includegraphics[width=0.65\columnwidth]{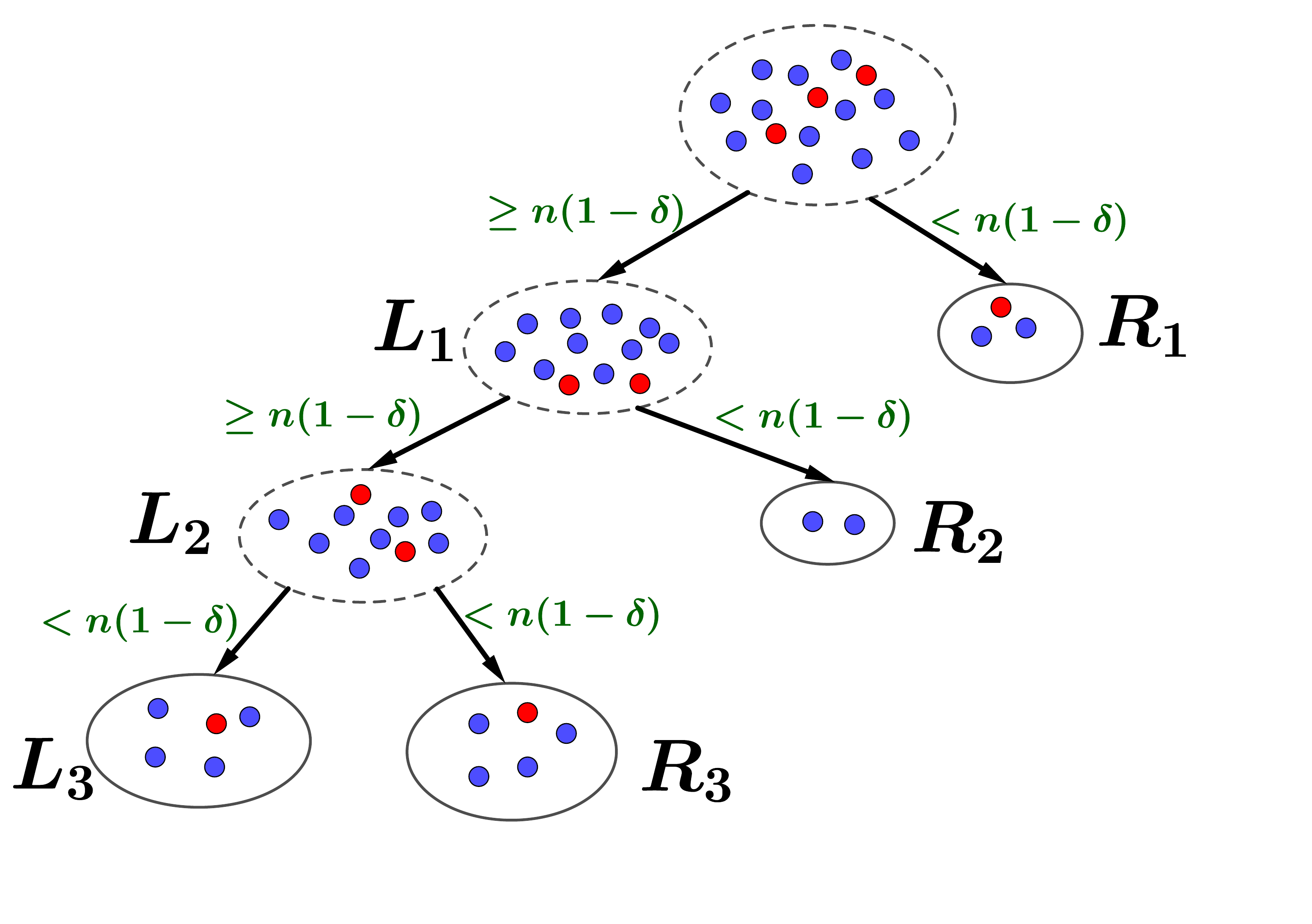}
     \caption{The layered structure of the optimum tree $T^*$ in Case 2.}
     \label{fig:opt-structure}
\end{figure}

 \begin{table}

\centering
 \begin{tabular}{||c|c||} 

  \hline
$\OPT$  & value of HC for the optimum solution $T^*$ \\ [0.5ex] 
 \hline
 $\optpeel$ & contribution of edges with at least one red endpoint in optimum  \\

 \hline

 $\optnpeel$ & contribution of blue edges in optimum \\
 \hline
 $\optnpeelchain$ &contribution of blue edges $(u,\cdot), u\in V\setminus(L_k\cup R_k)$ in optimum \\
 \hline
 $\optnpeelcut$ &contribution of blue edges $(u,v), u,v\in L_k\cup R_k$ in optimum \\
 \hline
 $\algpeel$ & objective value gained by our algorithm during peeling-off phase 1 \\
 \hline
 $\algnpeel$ &  objective value gained by our algorithm during max-cutting phase 2 \\
 \hline
$ \MC_{blue}$ & max-cut value only among blue vertices available to our algorithm in phase 2  \\
  \hline
\end{tabular}
\caption{A guide through the different variable names used in the proof.}
\label{tab:explain}
\end{table}

Before delving into the proof details for~\Cref{thm:vincent-upperbound}, we first provide an overview of the proof highlighting the main ideas.
\paragraph{Proof Sketch:} Our algorithm runs in two phases, namely the \emph{peeling-off phase} and the \emph{max-cutting phase} and a list of the symbols involved in the proof is provided in Table~\ref{tab:explain}. \\

\textbf{Step 1:} Even though our algorithm  removes vertices one by one during the peeling-off phase while the optimum tree $T^*$ removes chunks of nodes (i.e. the $R_i$'s), we will be flexible to ignore their contributions to $\OPT$ by only losing a small factor in the approximation, because these pieces are small.

\textbf{Step 2:} We want to devise a charging scheme between $\OPT$ and $\algtotal$. To achieve this, we further divide $\algtotal = \algpeel+\algnpeel$ (these are the contributions to the HC objective during the peeling-off and max-cutting phase respectively) and $\OPT=\optpeel+\optnpeel$. Suppose we mark the vertices that the algorithm peels off during the first phase as ``red'' and the rest of the vertices are marked as ``blue''.

\textbf{Step 3:} To take care of $\optpeel$ (this is the total contribution of edges with at least one red endpoint in the objective value of $T^*$), we only use $\algpeel$. Note that there can't be many high weighted degree vertices, so every vertex removed by $\algpeel$ had a significant multiplier in the HC objective. Since, $\optpeel$ could only have a multiplier of $n$ we get that $\algpeel\ge(1-\tfrac{2\tw}{n\tau})\optpeel$ (see~\Cref{lem:peel}). From this point on, we can completely ignore red vertices in the analysis.

\textbf{Step 4:} The remaining edges have blue both of their end-points (referred to as blue edges from now on). Let $\optnpeel$ be the total contribution of blue edges in the optimum $T^*$. Dealing with $\optnpeel$ requires more work. We need to further break $\optnpeel$ into: $\optnpeel=\optnpeelchain+\optnpeelcut$ (the total contribution of all blue edges with at least one end-point in $V\setminus(L_k\cup R_k)$ and the total contribution of all blue edges with both end-points in $L_k\cup R_k$ respectively). Note that $\optnpeelchain$ is negligible because it refers to low weighted degree vertices in small pieces, so $\optnpeelchain$ is a small fraction of $n\tw$ (and hence of $\OPT$ which is close to $n\tw$). See ~\Cref{lem:bluech}. 

\textbf{Step 5:} Finally, we will use $\algnpeel$ to take care of the $\optnpeelcut$ entirely. Actually, $\algnpeel$ will take care not only for the $\optnpeelcut$ (for now ignore some contribution from $w(L_k), w(R_k)$  because it is really small), but also at least half of the $\optnpeelchain$. See~\Cref{lem:maxcut}.

The above steps lead us to~\Cref{lem:final} which finishes the proof of~\Cref{thm:vincent-upperbound}.

 
 \begin{lemma}
 \label{lem:peel}
 Let $\ell$ denote the number of vertices peeled off during the first phase. Then $\ell\le\tfrac{2\tw}{\tau}$ and also $\algpeel\ge(1-\tfrac{2\tw}{n\tau})\optpeel$.
 \end{lemma}

 \begin{proof}
Every vertex that is peeled off during the first phase has weighted degree at least $\tau$. Observe that $\ell$ cannot be larger than $\tfrac{2\tw}{\tau}$, simply because the total sum of the weighted degrees is at most $2\tw$. Moreover, every peeled-off vertex $u$ (these are exactly the red vertices) belongs to a cluster of size at least $(n-\ell)\ge (n-\tfrac{2\tw}{\tau})$, hence $u$'s contribution to $\algpeel$ is at least $(n-\tfrac{2\tw}{\tau})\sum_{v\in V}w_{uv}$. Note that by the definition of $\optpeel$ we have:
\[
\optpeel\le n\cdot \sum_{u\ is\ red}\ \ \sum_{v\in V}w_{uv}
\]
Summing up the contributions of red vertices to $\algpeel$, we conclude the second part of the claim:
\begin{equation*}
\algpeel\ge(n-\tfrac{2\tw}{\tau})\sum_{u\ is\ red}\ \ \sum_{v\in V}w_{uv}\ge(1-\tfrac{2\tw}{n\tau})\optpeel
\qedhere
\end{equation*}
\end{proof}
 
We just obtained an upper bound for $\optpeel$ in terms of our algorithm's $\algpeel$ so we can ignore from now on the red vertices and turn our attention to $\optnpeel=\optnpeelchain + \optnpeelcut$. The first step is to upper bound $\optnpeelchain$:

\begin{lemma}
\label{lem:bluech}
$\optnpeelchain\le \delta\tau n^2\le \tfrac{2\delta \gamma}{1-\epsilon} \OPT$. 
\end{lemma}
\begin{proof}
As noted previously, $\rvert V\setminus(L_k\cup R_k)\rvert \le \delta n$, hence there are not that many vertices in $\rvert V\setminus(L_k\cup R_k)\rvert$. Since any edge that contributes to $\optnpeelchain$, must have, by definition, at least one endpoint in $\rvert V\setminus(L_k\cup R_k)\rvert$, there are at most $\delta n$ such edges and because they are blue, again by definition, their weighted degree is smaller than $\tau$. Noting that the maximum cluster size is at most $n$, we conclude that $\optnpeelchain\le (\delta n) \cdot \tau \cdot n = \delta\tau n^2$ and substituting $\tau$ in terms of $\gamma$, we get the lemma. 
\end{proof}
 
 
Let $w(L_k,R_k)$ be the the total weight of blue edges crossing the cut $(L_k,R_k)$ and let $w(R_k)$ and $w(L_k)$ be the total weight of the edges with both end-points in $R_k$ and $L_k$ respectively. An obvious upper bound that can be derived for $\optnpeelcut$ (recall that this refers only to blue edges), by focusing on the graph induced by the blue vertices in $L_k,R_k$, is $\optnpeelcut\le n (w(L_k,R_k) + w(L_k)+w(R_k))$.

 
After the max-cutting phase is over, we have no further control over the contribution of edges with both end-points in $L_k$ or both end-points in $R_k$, so we should better have an upper bound for their total weights. Informally, since $\OPT$ is large ({\textbf{Case 2}}) and both $L_k, R_k$ have small size, it can't be the case that significant portion of the weight lies inside $L_k,R_k$, as otherwise $\OPT$ would have to be small (the formal proof is deferred to the~\Cref{appendix:sec5}).
\begin{claim}
\label{cl:wLR}
$nw(L_k)+nw(R_k)\le \frac{\epsilon}{\delta}n\tw\le\tfrac{\epsilon}{(1-\epsilon)\delta}\OPT$.
\end{claim}
 
The next lemma starts by a lower bound for the $\MC_{blue}$ value, which is the value of the maximum cut in the graph induced only from the blue vertices available to our algorithm during its max-cutting phase, i.e. after we have removed the red vertices. Our algorithm will of course get only a $\rho_{\textrm{GW}}$-approximation ($\rho_{\textrm{GW}}\approx 0.878$) to $\MC_{blue}$, since it uses Goemans-Williamson for max-cut.
 
 \begin{lemma}
 \label{lem:maxcut}
$\algnpeel\ge \rho_{\textrm{GW}}\left(1-\tfrac{2\tw}{n\tau}\right)\left(\optnpeelcut-\tfrac{\epsilon}{(1-\epsilon)\delta}\OPT + \tfrac{\optnpeelchain}{2}\right)$
 \end{lemma}
 \begin{proof}
As mentioned, we know that $ n w(L_k,R_k)\ge\optnpeelcut  - nw(L_k)-nw(R_k)\ge \optnpeelcut - \tfrac{\epsilon}{(1-\epsilon)\delta}\OPT$. During the max-cutting phase, the vertices available to our algorithm are all the blue vertices. These can be divided into two categories relative to the $\OPT$ solution. The first category are blue vertices $u\in L_k\cup R_k$. The second category are blue vertices $v\in V\setminus(L_k\cup R_k)$. Imagine the following cut $(C,\bar{C})$ with weight $w(C,\bar{C})$: focus only  on the vertices $u$ of the first category and split them into two pieces optimally to obtain the maximum cut. Now randomly assign the vertices $v$ of the second category to the two pieces. This cut $(C,\bar{C})$ would obtain, by definition, an HC objective value of:
\[
n\cdot w(C,\bar{C})\ge n\cdot w(L_k,R_k) +  \frac{\optnpeelchain}{2}\ge\optnpeelcut  - \tfrac{\epsilon}{(1-\epsilon)\delta}\OPT+ \frac{\optnpeelchain}{2}
\]
Since $\MC_{blue}$ is the optimal cut, it can only be better than $w(C,\bar{C})$ and hence:
\[
n \cdot \MC_{blue}\ge \optnpeelcut-\tfrac{\epsilon}{(1-\epsilon)\delta}\OPT + \frac{\optnpeelchain}{2}
\] 
We know from~\Cref{lem:peel}, that at the beginning of the max-cutting phase, our algorithm has removed at most $\tfrac{2\tw}{\tau}$ vertices and hence, the cluster size at the point where our algorithm uses the Goemans-Williamson algorithm is at least $(n-\frac{2\tw}{\tau})$. The lemma follows since we can only get a $\rho_{\textrm{GW}}$ approximation to $\MC_{blue}$.
\end{proof}
Finally, we are able to combine all the above together into the final comparison between our algorithm's objective value $\algtotal$ and the optimum $\OPT$:

\begin{lemma}
\label{lem:final}
Let $\tau=\gamma \tfrac{2\tw}{n}, \delta=\tfrac{\sqrt{\epsilon}}{\sqrt{\gamma}}$. By optimizing for the parameters $\gamma, \epsilon$, we get an $\alpha$-approximation to the dissimilarity-HC objective, where $\alpha=0.667078 > \tfrac23$.
\end{lemma}
\begin{proof}
The proof involves optimizing for the parameters $\epsilon,\gamma,\delta$ and balancing out the two factors obtained from \textbf{Case 1} and \textbf{Case 2}. The final equation is:
\[
\rho_{\textrm{GW}}\left(1-\tfrac{1}{\gamma}\right)\left(1-\tfrac{\epsilon/\delta}{1-\epsilon}-\tfrac{\delta\gamma}{1-\epsilon}\right)=\tfrac{2}{3(1-\epsilon)}
\]
We defer the details of this proof to the~\Cref{appendix:sec5}. This finishes the proof of~\Cref{thm:vincent-upperbound}.
\end{proof}

\section{Conclusion - Discussion}
\label{sec:conclusion}
In this paper, we design algorithms for hierarchical clustering that perform better than the method of choice in practice, which is Average-Linkage. We view our work as part of the ongoing effort to understand what a good HC objective would be and we believe that further research is required to understand this question to the point where objectives would lead to algorithms with high-quality outcomes as it has been the case for flat-clustering with the $k$-means, $k$-median and $k$-center objectives. 

The reason why defining a good HC objective seems hard is that in a hierarchical tree all edges get cut eventually so the crucial decision to be made is how to penalize the cuts of different solutions. A common characteristic for all three HC objectives presented thus far is that an edge $w_{ij}$ is penalized according to the fraction of datapoints present at the moment when the $i,j$ are separated. This leads to desirable properties as mentioned in the introduction but there may be other ways to go as well. Ideally, an objective that truly differentiates Average-Linkage or realizing what properties of real-world data would allow Average-Linkage to perform much better than the worst-case or better than other ad-hoc algorithms (and random solutions) are steps towards the right direction.

\bibliographystyle{plainnat}

\bibliography{refs}

\begin{thebibliography}{28}
\providecommand{\natexlab}[1]{#1}
\providecommand{\url}[1]{\texttt{#1}}
\expandafter\ifx\csname urlstyle\endcsname\relax
  \providecommand{\doi}[1]{doi: #1}\else
  \providecommand{\doi}{doi: \begingroup \urlstyle{rm}\Url}\fi

\bibitem[Awasthi et~al.(2014)Awasthi, Balcan, and Voevodski]{nina2}
Pranjal Awasthi, Maria Balcan, and Konstantin Voevodski.
\newblock Local algorithms for interactive clustering.
\newblock In \emph{International Conference on Machine Learning}, pages
  550--558, 2014.

\bibitem[Balcan and Blum(2008)]{nina1}
Maria-Florina Balcan and Avrim Blum.
\newblock Clustering with interactive feedback.
\newblock In \emph{International Conference on Algorithmic Learning Theory},
  pages 316--328. Springer, 2008.

\bibitem[Berkhin(2006)]{berkhin2006survey}
Pavel Berkhin.
\newblock A survey of clustering data mining techniques.
\newblock In \emph{Grouping multidimensional data}, pages 25--71. Springer,
  2006.

\bibitem[Charikar and Chatziafratis(2017)]{vaggosSODA}
Moses Charikar and Vaggos Chatziafratis.
\newblock Approximate hierarchical clustering via sparsest cut and spreading
  metrics.
\newblock In \emph{Proceedings of the Twenty-Eighth Annual ACM-SIAM Symposium
  on Discrete Algorithms}, pages 841--854. Society for Industrial and Applied
  Mathematics, 2017.

\bibitem[Charikar et~al.(2004)Charikar, Chekuri, Feder, and
  Motwani]{charikar2004incremental}
Moses Charikar, Chandra Chekuri, Tom{\'a}s Feder, and Rajeev Motwani.
\newblock Incremental clustering and dynamic information retrieval.
\newblock \emph{SIAM Journal on Computing}, 33\penalty0 (6):\penalty0
  1417--1440, 2004.

\bibitem[Chatziafratis et~al.(2018{\natexlab{a}})Chatziafratis, Niazadeh, and
  Charikar]{chatziafratis2018hierarchical}
Vaggos Chatziafratis, Rad Niazadeh, and Moses Charikar.
\newblock Hierarchical clustering with structural constraints.
\newblock \emph{arXiv preprint arXiv:1805.09476}, 2018{\natexlab{a}}.
\newblock URL \url{https://arxiv.org/abs/1805.09476}.

\bibitem[Chatziafratis et~al.(2018{\natexlab{b}})Chatziafratis, Niazadeh, and
  Charikar]{vaggosICML}
Vaggos Chatziafratis, Rad Niazadeh, and Moses Charikar.
\newblock Hierarchical clustering with structural constraints.
\newblock In \emph{International Conference on Machine Learning}, pages
  773--782, 2018{\natexlab{b}}.

\bibitem[Cohen-Addad(2018)]{vincentpersonal}
Vincent Cohen-Addad.
\newblock \emph{Personal Communication}, 2018.

\bibitem[Cohen-Addad et~al.(2017)Cohen-Addad, Kanade, and
  Mallmann-Trenn]{vincentNIPS}
Vincent Cohen-Addad, Varun Kanade, and Frederik Mallmann-Trenn.
\newblock Hierarchical clustering beyond the worst-case.
\newblock In \emph{Advances in Neural Information Processing Systems}, pages
  6202--6210, 2017.

\bibitem[Cohen-Addad et~al.(2018)Cohen-Addad, Kanade, Mallmann-Trenn, and
  Mathieu]{vincentSODA}
Vincent Cohen-Addad, Varun Kanade, Frederik Mallmann-Trenn, and Claire Mathieu.
\newblock Hierarchical clustering: Objective functions and algorithms.
\newblock In \emph{Proceedings of the Twenty-Ninth Annual ACM-SIAM Symposium on
  Discrete Algorithms}, pages 378--397. SIAM, 2018.

\bibitem[Dasgupta(2002)]{dasgupta2002performance}
Sanjoy Dasgupta.
\newblock Performance guarantees for hierarchical clustering.
\newblock In \emph{International Conference on Computational Learning Theory},
  pages 351--363. Springer, 2002.

\bibitem[Dasgupta(2016)]{dasguptaSTOC}
Sanjoy Dasgupta.
\newblock A cost function for similarity-based hierarchical clustering.
\newblock In \emph{Proceedings of the forty-eighth annual ACM symposium on
  Theory of Computing}, pages 118--127. ACM, 2016.

\bibitem[Diez et~al.(2015)Diez, Bonifazi, Escudero, Mateos, Mu{\~n}oz,
  Stramaglia, and Cortes]{diez2015novel}
Ibai Diez, Paolo Bonifazi, I{\~n}aki Escudero, Beatriz Mateos, Miguel~A
  Mu{\~n}oz, Sebastiano Stramaglia, and Jesus~M Cortes.
\newblock A novel brain partition highlights the modular skeleton shared by
  structure and function.
\newblock \emph{Scientific reports}, 5:\penalty0 10532, 2015.

\bibitem[Eisen et~al.(1998)Eisen, Spellman, Brown, and
  Botstein]{eisen1998cluster}
Michael~B Eisen, Paul~T Spellman, Patrick~O Brown, and David Botstein.
\newblock Cluster analysis and display of genome-wide expression patterns.
\newblock \emph{Proceedings of the National Academy of Sciences}, 95\penalty0
  (25):\penalty0 14863--14868, 1998.

\bibitem[Friedman et~al.(2001)Friedman, Hastie, and
  Tibshirani]{friedman2001elements}
Jerome Friedman, Trevor Hastie, and Robert Tibshirani.
\newblock \emph{The elements of statistical learning}, volume~1.
\newblock Springer series in statistics New York, NY, USA:, 2001.

\bibitem[Goemans and Williamson(1995)]{goemans1995improved}
Michel~X Goemans and David~P Williamson.
\newblock Improved approximation algorithms for maximum cut and satisfiability
  problems using semidefinite programming.
\newblock \emph{Journal of the ACM (JACM)}, 42\penalty0 (6):\penalty0
  1115--1145, 1995.

\bibitem[Jardine and Sibson(1968)]{jardine1968model}
N~Jardine and R~Sibson.
\newblock A model for taxonomy.
\newblock \emph{Mathematical Biosciences}, 2\penalty0 (3-4):\penalty0 465--482,
  1968.

\bibitem[Leskovec et~al.(2014)Leskovec, Rajaraman, and
  Ullman]{leskovec2014mining}
Jure Leskovec, Anand Rajaraman, and Jeffrey~David Ullman.
\newblock \emph{Mining of massive datasets}.
\newblock Cambridge university press, 2014.

\bibitem[Lin et~al.(2010)Lin, Nagarajan, Rajaraman, and
  Williamson]{lin2010general}
Guolong Lin, Chandrashekhar Nagarajan, Rajmohan Rajaraman, and David~P
  Williamson.
\newblock A general approach for incremental approximation and hierarchical
  clustering.
\newblock \emph{SIAM Journal on Computing}, 39\penalty0 (8):\penalty0
  3633--3669, 2010.

\bibitem[Mann et~al.(2008)Mann, Matula, and Olinick]{mann2008use}
Charles~F Mann, David~W Matula, and Eli~V Olinick.
\newblock The use of sparsest cuts to reveal the hierarchical community
  structure of social networks.
\newblock \emph{Social Networks}, 30\penalty0 (3):\penalty0 223--234, 2008.

\bibitem[McSherry(2001)]{mcsherry2001spectral}
Frank McSherry.
\newblock Spectral partitioning of random graphs.
\newblock In \emph{focs}, page 529. IEEE, 2001.

\bibitem[Moseley and Wang(2017)]{joshNIPS}
Benjamin Moseley and Joshua Wang.
\newblock Approximation bounds for hierarchical clustering: Average linkage,
  bisecting k-means, and local search.
\newblock In \emph{Advances in Neural Information Processing Systems}, pages
  3097--3106, 2017.

\bibitem[Plaxton(2003)]{plaxton2003approximation}
C~Greg Plaxton.
\newblock Approximation algorithms for hierarchical location problems.
\newblock In \emph{Proceedings of the thirty-fifth annual ACM symposium on
  Theory of computing}, pages 40--49. ACM, 2003.

\bibitem[Roy and Pokutta(2016)]{royNIPS}
Aurko Roy and Sebastian Pokutta.
\newblock Hierarchical clustering via spreading metrics.
\newblock In \emph{Advances in Neural Information Processing Systems}, pages
  2316--2324, 2016.

\bibitem[Sneath and Sokal(1962)]{sneath1962numerical}
Peter~HA Sneath and Robert~R Sokal.
\newblock Numerical taxonomy.
\newblock \emph{Nature}, 193\penalty0 (4818):\penalty0 855--860, 1962.

\bibitem[Steinbach et~al.(2000)Steinbach, Karypis, Kumar,
  et~al.]{steinbach2000comparison}
Michael Steinbach, George Karypis, Vipin Kumar, et~al.
\newblock A comparison of document clustering techniques.
\newblock In \emph{KDD workshop on text mining}, volume 400. Boston, 2000.

\bibitem[Tumminello et~al.(2010)Tumminello, Lillo, and
  Mantegna]{tumminello2010correlation}
Michele Tumminello, Fabrizio Lillo, and Rosario~N Mantegna.
\newblock Correlation, hierarchies, and networks in financial markets.
\newblock \emph{Journal of Economic Behavior \& Organization}, 75\penalty0
  (1):\penalty0 40--58, 2010.

\bibitem[Vikram and Dasgupta(2016)]{dasguptaICML}
Sharad Vikram and Sanjoy Dasgupta.
\newblock Interactive bayesian hierarchical clustering.
\newblock In \emph{International Conference on Machine Learning}, pages
  2081--2090, 2016.

\end{thebibliography}
\appendix
\section{Deferred Proofs of \texorpdfstring{\Cref{sec:tight-similarity}}{TEXT}}
\label{appendix:sec3}
\begin{proof}[Proof of \Cref{lem:tight-sim-1}]

The bottom-up merging strategy that first merges the edges $e$ inside the $K_{n^{2/3}}$ cliques attains HC value close to $n\tw$ (the rest of the edges don't matter). Since $\OPT$ is only better than this merging strategy, the claim follows. To see that, note that all such edges $e$ will have a multiplier of $\nl\ge n-n^{2/3}$. Since there are $\tfrac{1}{2}n^{2/3}\cdot (n^{2/3}-1)\cdot n^{1/3}$ such edges, $\OPT\ge (n-n^{2/3})\cdot\tfrac{1}{2}n^{2/3}\cdot (n^{2/3}-1)\cdot n^{1/3}\geq \tfrac12 n^{8/3} -O(n^{7/3}).$

\end{proof}

\begin{proof}[Proof of \Cref{lem:tight-sim-2}]
Because of the $1+\epsilon$ weight of the edges going across the $K_{n^{2/3}}$ cliques, $\AL$ will first start merging the $K_{n^{1/3}}$ cliques consisting of one node out of each $K_{n^{2/3}}$ clique. There are $n^{2/3}$ such $K_{n^{1/3}}$ cliques so the total objective contribution of the edges involved in this first phase is insignificant since it is certainly smaller than $n\cdot \tfrac{1}{2}n^{1/3}\cdot (n^{1/3}-1)\cdot n^{2/3} \cdot (1+\epsilon) = O(n^{7/3})$.
Observe that after the first phase of $\AL$ the remaining subclusters to be merged form a clique on $n^{2/3}$ supernodes each with size $s=n^{1/3}$ and weighted edges with uniform weights $n^{1/3}$. By~\Cref{obs:table-vals} and taking into account the size of every supernode, we obtain the final value for $\AL= \tfrac13 n^{2/3}\cdot{n^{2/3}\choose2}\cdot w\cdot s = \tfrac13n^{2/3}\cdot \tfrac12 n^{2/3}\cdot(n^{2/3}-1)\cdot n^{1/3}\cdot n^{1/3} \leq \tfrac16 n^{8/3} +O(n^{7/3})$.
\end{proof}
\begin{proof}[Proof of \Cref{lem:tight-disim-1}]
The $\OPT$ solution can get all the weight by performing the cut $(L,R)$ and then proceed arbitrarily. The HC value is then $\OPT=n\tw = \tfrac14 n^{3} - O(n^2)$.
\end{proof}
\begin{proof}[Proof of \Cref{lem:tight-disim-2}]
Since there are a lot of 0 weight edges in the graph, $\AL$ first tries to merge endpoints of such edges. Note that $\AL$ is underspecified since there are ties here, but these ties are not affecting the overall outcome as we can break ties arbitrarily by using small edge weights $\epsilon >0$. Hence, we can assume that $\AL$ first merged the two endpoints of edges in the perfect matching $M$. After this first step, the remaining subclusters to be merged form a clique on $\tfrac n2$ supernodes each with size $s=2$ and weighted edges with uniform weights $w=2$. By using the~\Cref{obs:table-vals} and taking into account the size of every supernode, we obtain the final value for $\AL= \tfrac23 \tfrac n2\cdot{\binom{n/2}{2}}\cdot w\cdot s = \tfrac23 \tfrac n2\cdot\tfrac 12 \cdot \tfrac n2 \cdot(\tfrac n2 -1)\cdot 2\cdot 2 
\leq \tfrac 16n^3$.
\end{proof}
\section{Deferred Proofs of \texorpdfstring{\Cref{sec:beating-similarity}}{TEXT}}
\label{appendix:sec4}

\begin{proof}[Proof of \Cref{lem:factor-revealing}] 
First of all, the optimization \ref{eq:factor-reveal} decomposes over variables $\{\theta_{ik}\}_{k\neq i}$ and variables $\{\theta_{jk}\}_{j\neq k}$. Due to symmetry, we only lower-bound the optimal objective value of the following minimization program,:
\begin{equation}
\label{eq:factor-reveal-2}
\tag{P-1}
\begin{aligned}
& \underset{}{\textit{minimize}}
& &\sum_{k\neq i,j}  \theta_{ik} \\
& \textit{subject to}
& &\sum_{k\neq i}\cos(\theta_{ik})\leq n/2-1, & \\
& 
& &0\leq \theta_{ik}\leq\frac{\pi}{2},&\forall k,\\
\end{aligned}
\end{equation}
and then use \texttt{OBJ}(\ref{eq:factor-reveal})=$\frac{1}{2\pi}\big($2 \texttt{OBJ}(\ref{eq:factor-reveal-2})$-(n-2)\bar{\theta}\big)$. Suppose $\{\theta^*_{ik}\}_{k\neq i}$ is the optimal solution of the above program (\ref{eq:factor-reveal-2}). We first claim that in any optimal solution the first constraint is tight, i.e. $\sum_{k\neq i}\cos(\theta^*_{ik})=n/2-1$. This simply holds because otherwise one can slightly decrease one of the non-zero $\theta^*_{ik}$ and strictly decrease the objective, a contradiction. Next, we claim that $\theta^*_{ik}\in\{0,\frac{\pi}{2}\}$ for all $k\neq i$, except for at most one $k=k_0$. To prove by contradiction, suppose it is not true. Therefore, there exist $k_1,k_2\neq i$ such that $ 0<\theta^*_{ik_1}\leq \theta^*_{ik_2}<\pi/2$. If we decrease $\theta^*_{ik_1}$ by infinitesimal $d_\theta$ and increase $\theta^*_{ik_2}$ by the same $d_\theta$, then the objective value does not change. However, because of the concavity of the cosine function over the interval $[0,\pi/2]$, there will be an additional slack in the first constraint of \ref{eq:factor-reveal-2}, a contradiction to the first claim that in any optimal solution this constraint is tight.

Because $\theta^*_{ik}\in\{0,\pi/2\}$ for $k\neq i,k_0$,  we have:
\begin{equation*}
\#\{k\neq i,k_0: \theta^*_{ik}=0\}=\sum_{k\neq i,k_0}\cos(\theta^*_{i,k})\leq \sum_{k\neq i}\cos(\theta^*_{i,k})\leq n/2-1
\end{equation*}
We then conclude that for at least $n-2-(n/2-1)=\frac{n-2}{2}$ values of $k\neq i,k_0$ we have $\theta^*_{ik}=\pi/2$, and hence the optimal objective value of \ref{eq:factor-reveal-2} is lower-bounded by $\frac{n-2}{2}\cdot\pi/2=\frac{(n-2)\pi}{4}$. This lower-bound immediately implies that the optimal objective value of \ref{eq:factor-reveal} is also lower-bounded by  
$$\frac{2\cdot \frac{(n-2)\pi}{4}-(n-2)\bar{\theta}}{2\pi}=(n-2)\left(\frac{1}{4}-\frac{\bar{\theta}}{2\pi}\right),$$
which completes the proof of the lemma.
\end{proof}

\begin{proof}[Details of final calculations in the proof of \texorpdfstring{\Cref{thm:similarity}}{TEXT}]
To get the final approximation factor, we balanced out the two cases:
\begin{equation}
\left(1-\frac{2\epsilon_1}{\epsilon_2}\right)\left(\frac{1}{2}-\frac{2\cos^{-1}(1-\epsilon_2)}{3\pi}\right)=\frac{1}{3(1-\epsilon_1)}
\end{equation}
By solving for $\epsilon_1$, the optimal value of $\epsilon_1$ as a function of $\epsilon_2$ is calculated to be the following function:
\begin{equation*}
\epsilon^*_1(\epsilon_2)=\frac{1}{4}\cdot\left(\epsilon_2+2)-\sqrt{(\epsilon_2+2)^2-8\epsilon_2\left(1-\frac{1}{3\cdot(\frac{1}{2}-\frac{2\arccos{1-\epsilon_2}}{3\pi})}\right)}\right)
\end{equation*}
we then draw $\alpha(\epsilon_2)=\frac{1}{3(1-\epsilon^*_1(\epsilon_2))}$ for $\epsilon_2\in [0,1]$ by the aid of a computer software (\texttt{WolframAlpha}). This function peaks at around $\epsilon_2\approx 0.139$. By plugging this number into $\alpha(\epsilon_2)$, we get the final factor.
\end{proof}

\section{Deferred Proofs of \texorpdfstring{\Cref{sec:beating-dissimilarity}}{TEXT}}
\label{appendix:sec5}

\begin{proof}[Proof of \Cref{cl:wLR}]
 Recall that in the partition of $T^*$, we have $\max\left(\lvert L_k\rvert, \lvert R_k\rvert\right)< n(1-\delta)$. Since both pieces $L_k, R_k$ have sizes at most $(1-\delta)n$, edges cut within the subtrees rooted at $L_k$ and at $R_k$ can only have contribution to $\OPT$ at most $n(1-\delta)w(L_k)+n(1-\delta)w(R_k)$. Hence:
\[
\OPT\le n(1-\delta)w(L_k)+n(1-\delta)w(R_k) + n(\tw-w(L_k)-w(R_k))
\] 
Combining with our assumption $\OPT \geq (1-\epsilon)n\tw$, we get:
\[
(1-\epsilon)n\tw\le n(1-\delta)w(L_k)+n(1-\delta)w(R_k) + n(\tw-w(L_k)-w(R_k))
\]
and the claim follows by rearranging the terms after the cancelations.
\end{proof}

\begin{proof}[Proof of \Cref{lem:final}] 
We start by $\algtotal=\algpeel+\algnpeel$. From~\Cref{lem:peel} and Lemma~\ref{lem:maxcut} we have:
\[
\algpeel\ge(1-\tfrac{2\tw}{n\tau})\optpeel
\]
\[ \algnpeel\ge \rho_{\textrm{GW}} \left(n-\tfrac{2\tw}{\tau}\right) \MC_{blue} \ge  \rho_{\textrm{GW}}\left(1-\tfrac{2\tw}{n\tau}\right)\left(\optnpeelcut-nw(L_k)-nw(R_k) + \tfrac{\optnpeelchain}{2}\right)
\]
Hence, 
\[
\algtotal\ge \rho_{\textrm{GW}}\left(1-\tfrac{2\tw}{n\tau}\right)\left(\optpeel+\optnpeelcut + \optnpeelchain-nw(L_k) - nw(R_k) - \tfrac{\optnpeelchain}{2}\right)
\]
Because $\OPT\ge (1-\epsilon)n\tw \implies n\tw\le \tfrac{\OPT}{1-\epsilon}$, by Claim~\ref{cl:wLR} and~\Cref{lem:bluech} we get ($\delta\tau n^2=2\delta\gamma n\tw$):
\[
\algtotal\ge \rho_{\textrm{GW}}\left(1-\tfrac{1}{\gamma}\right)\left(1-\tfrac{\epsilon}{\delta(1-\epsilon)}-\tfrac{\delta\gamma}{1-\epsilon}\right)\OPT
\]
We have to balance out the two factors obtained from \textbf{Case 1} and \textbf{Case 2}, so we get the final equation:
\[
\rho_{\textrm{GW}}\left(1-\tfrac{1}{\gamma}\right)\left(1-\tfrac{\epsilon/\delta}{1-\epsilon}-\tfrac{\delta\gamma}{1-\epsilon}\right)=\tfrac{2}{3(1-\epsilon)}
\]
In terms of the parameter $\delta$, it's easy to see that the choice of $\delta=\sqrt{\tfrac{\epsilon}{\gamma}}$ is optimal, so substituting:
\[
\rho_{\textrm{GW}}\left(1-\tfrac{1}{\gamma}\right)\left(1-\tfrac{2\sqrt{\epsilon\gamma}}{1-\epsilon}\right)=\tfrac{2}{3(1-\epsilon)}
\]
Rearranging the terms we get:
\[
\rho_{\textrm{GW}}\left(1-\tfrac{1}{\gamma}\right) \epsilon + 2\rho_{\textrm{GW}}\left(1-\tfrac{1}{\gamma}\right)\sqrt{\gamma} \cdot \sqrt{\epsilon} + \tfrac{2}{3}-\rho_{\textrm{GW}}\left(1-\tfrac{1}{\gamma}\right) = 0
\]
Maximizing over $\gamma$ for $\sqrt{\epsilon}\ge0$ (dropping the negative solution):
\[
\sqrt{\epsilon}= \dfrac{-2\sqrt{\gamma} + \sqrt{4\gamma-4\left(\tfrac{2}{3}\tfrac{1}{\rho_{\textrm{GW}}}\tfrac{\gamma}{\gamma-1}-1\right)}}{2}
\]
we get the final optimal answer for $\gamma\approx 11.1$ and $\epsilon\approx 0.000612$. 

\end{proof}

\end{document}